\documentclass[10pt,A4paper,conference]{IEEEtran}
\usepackage{amsthm}
\usepackage{amsmath}
\usepackage{amsfonts}
\usepackage{graphicx}
\usepackage{latexsym}
\usepackage{amssymb}
\usepackage{stmaryrd}
\usepackage{amscd}
\usepackage[small]{caption}

\newcommand{\deff}{\mbox{$\stackrel{\rm def}{=}$}}

\newcommand{\sbinom}[2]{\left[ \begin{array}{c} #1 \\ #2 \end{array} \right] }

\newcommand{\field}[1]{\mathbb{#1}}

\newcommand{\cC}{{\cal C}}

\newcommand{\cM}{{\cal M}}

\newcommand{\sP}{\field{P}}

\newcommand{\sG}{\field{G}}

\DeclareMathAlphabet{\mathbfsl}{OT1}{cmr}{bx}{it}
\newcommand{\uuu}{\kern-1pt\mathbfsl{u}\kern-0.5pt}
\newcommand{\vvv}{\kern-1pt\mathbfsl{v}\kern-0.5pt}

\newcommand{\myboxplus}{\kern1pt\mbox{\small$\boxplus$}}

\makeatletter \DeclareRobustCommand{\sbinom}{\genfrac[]\z@{}}
\makeatother
\newcommand{\G}[2]{\sbinom{{#1}\kern-1pt}{{#2}\kern-1pt}}
\newcommand{\Gq}[2]{\sbinom{{#1}\kern-0.25pt}{{#2}\kern-0.25pt}}

\newcommand{\Ps}{\smash{{\sP\kern-2.0pt}_q\kern-0.5pt(n)}}
\newcommand{\sPs}{\smash{{\sP\kern-1.5pt}_q(n)}}
\newcommand{\Ptwo}{\smash{{\sP\kern-2.0pt}_2\kern-0.5pt(n)}}
\newcommand{\Ptwom}{\smash{{\sP\kern-2.0pt}_2\kern-0.5pt(m)}}
\newcommand{\Ptwonm}{\smash{{\sP\kern-2.0pt}_2\kern-0.5pt(n+m)}}
\newcommand{\Ptwoa}{\smash{{\sP\kern-2.0pt}_2\kern-0.5pt(1)}}
\newcommand{\Ptwob}{\smash{{\sP\kern-2.0pt}_2\kern-0.5pt(2)}}
\newcommand{\Ptwoc}{\smash{{\sP\kern-2.0pt}_2\kern-0.5pt(3)}}
\newcommand{\Ptwod}{\smash{{\sP\kern-2.0pt}_2\kern-0.5pt(4)}}
\newcommand{\Ptwoe}{\smash{{\sP\kern-2.0pt}_2\kern-0.5pt(5)}}
\newcommand{\Ptwof}{\smash{{\sP\kern-2.0pt}_2\kern-0.5pt(6)}}
\newcommand{\Ptwokm}{\smash{{\sP\kern-2.0pt}_2\kern-0.5pt(2k-1)}}
\newcommand{\Pone}{\smash{{\sP\kern-2.5pt}_2\kern-0.5pt(n{-}1)}}

\newcommand{\Gr}{\smash{{\sG\kern-1.5pt}_q\kern-0.5pt(n,k)}}
\newcommand{\Gi}{\smash{{\sG\kern-1.5pt}_q\kern-0.5pt(n,i)}}
\newcommand{\Gj}{\smash{{\sG\kern-1.5pt}_q\kern-0.5pt(n,j)}}
\newcommand{\Grmk}{\smash{{\sG\kern-1.5pt}_q\kern-0.5pt(n,n-k)}}
\newcommand{\Grdk}{\smash{{\sG\kern-1.5pt}_q\kern-0.5pt(2k,k)}}
\newcommand{\Grekappa}{\smash{{\sG\kern-1.5pt}_q\kern-0.5pt(n,e+1-\kappa)}}
\newcommand{\Grtwoekappa}{\smash{{\sG\kern-1.5pt}_q\kern-0.5pt(n,2e+1-\kappa)}}
\newcommand{\Gremkappa}{\smash{{\sG\kern-1.5pt}_q\kern-0.5pt(n,e-\kappa)}}
\newcommand{\Gn}{\smash{{\sG\kern-1.5pt}_2\kern-0.5pt(n,n{-}1)}}
\newcommand{\Gnq}{\smash{{\sG\kern-1.5pt}_q\kern-0.5pt(n,n{-}1)}}
\newcommand{\Gone}{\smash{{\sG\kern-1.5pt}_2\kern-0.5pt(n,1)}}
\newcommand{\Gqone}{\smash{{\sG\kern-1.5pt}_q\kern-0.5pt(n,1)}}
\newcommand{\GTwo}{\smash{{\sG\kern-1.5pt}_2\kern-0.5pt(n,k)}}
\newcommand{\GTwonk}[2]{{\smash{{\sG\kern-1.5pt}_2\kern-0.5pt({#1},{#2})}}}
\newcommand{\Gnk}{\smash{{\sG\kern-1.5pt}_2\kern-0.5pt(n,n{-}k)}}
\newcommand{\Greone}{\smash{{\sG\kern-1.5pt}_q\kern-0.5pt(n,e{+}1)}}
\newcommand{\Gretwo}{\smash{{\sG\kern-1.5pt}_q\kern-0.5pt(n,e{+}2)}}

\newcommand{\be}[1]{\begin{equation}\label{#1}}
\newcommand{\ee}{\end{equation}}

\newcommand{\Cref}[1]{Co\-rol\-la\-ry\,\ref{#1}}

\newtheorem{theorem}{Theorem}
\newtheorem{lemma}{Lemma}
\newtheorem{remark}{Remark}
\newtheorem{cor}{Corollary}
\newtheorem{definition}{Definition}

\newtheorem{example}{Example}
\newtheorem{construction}{Construction}

\begin{document}

\title{\huge \textbf{Constrained Codes for Rank Modulation}\vspace{-1ex}}
\author{\authorblockN{\textbf{Sarit Buzaglo}}
\authorblockA{Department of Computer Science\\
Technion-Israel Institute of Technology\\
Haifa 32000, Israel \\
Email: \emph{sarahb@cs.technion.ac.il}\vspace*{-5.0ex}} \and
\authorblockN{\textbf{Eitan Yaakobi}}
\authorblockA{Electrical Engineering Department\\
California Institute of Technology\\
Pasadena, CA 91125, USA \\
Email: \emph{yaakobi@caltech.edu}\vspace*{-5.0ex}}}

\maketitle

\begin{abstract}
Motivated by the rank modulation scheme, a recent work by Sala and Dolecek explored the study of constraint codes for permutations. The constraint studied by them is inherited by the inter-cell interference phenomenon in flash memories, where high-level cells can inadvertently increase the level of low-level cells.

In this paper, the model studied by Sala and Dolecek is extended into two constraints. A permutation $\sigma \in S_n$ satisfies the \emph{two-neighbor $k$-constraint} if for all $2 \leq i \leq n-1$ either $|\sigma(i-1)-\sigma(i)|\leq k$ or $|\sigma(i)-\sigma(i+1)|\leq k$, and it satisfies the \emph{asymmetric two-neighbor $k$-constraint} if for all $2 \leq i \leq n-1$, either $\sigma(i-1)-\sigma(i) < k$ or $\sigma(i+1)-\sigma(i) < k$. We show that the capacity of the first constraint is $(1+\epsilon)/2$ in case that $k=\Theta(n^{\epsilon})$ and the capacity of the second constraint is 1 regardless to the value of $k$. We also extend our results and study the capacity of these two constraints combined with error-correction codes in the Kendall's $\tau$ metric.
\end{abstract}

\vspace{0ex}
\section{Introduction}
Flash memories are, by far, the most important type of non-volatile memory (NVM) in use today. Flash devices are employed widely in mobile, embedded, and mass-storage applications, and the growth in this sector continues at a staggering pace. At the high level, flash memories are comprised of block of cells. These cells can have binary values, i.e. they store a single bit, or can have multiple levels and thus can store multiple bits in a cell.

One of the main challenges in flash memories is to exactly program each cell to its designated level. Furthermore, flash memories suffer from the cell leakage problem, by which a charge may leak from the cells and thus cause reading errors~\cite{CM99}. In order to overcome these difficulties, the novel framework of \textbf{\emph{rank modulation codes}} was introduced in~\cite{JMSB09}. Under this setup, the information is represented by permutations which are derived by the relative charge levels of the cells, rather than by their absolute levels.

Another conspicuous property of flash memories, resulting from its rapid growth density, is the appearance of inter-cell interference (ICI). The level of a cell, called a {\em victim cell} might increase, if its neighbor cells are programmed to significantly higher levels~\cite{LHC02}. The ICI is caused by the parasitic capacitance between neighboring cells, and in particular, multilevel cell programming is severely influenced by this effect.

Motivated by the rank modulation scheme and the ICI phenomenon, a recent research by Sala and Dolecek~\cite{SaDo13} proposed the study of constraint codes for permutations. Under this setup, the constraint is invoked over the permutation's symbols. In the model studied in~\cite{SaDo13}, the authors explored the constraint in which the levels difference between adjacent cells is upper bounded. This constraint prevents the scenario in which a high-level cell affects its low-level neighbor cell. Namely, let $S_n$ be the set of all permutations with $n$ elements, then it was said that a permutation $\sigma \in S_n$ satisfies the \textbf{\emph{single-neighbor $k$-constraint}} if $|\sigma_i-\sigma_{i+1}|\leq k$ for all $1\leq i\leq n-1$. For example, the permutation $\sigma=[3,1,2,4,5]$ satisfies the single-neighbor 2-constraint but not the single-neighbor 1-constraint. For any positive integers $k$ and $n$, if $U_{n,k}$ is the set of permutations that meet the single-neighbor $k$-constraint, then the capacity of this constraint is defined as $C_1(\epsilon) = \lim_{n\rightarrow \infty} \frac{\log |U_{n,k}|}{\log n!}$, where $k=\left\lceil n^\epsilon\right\rceil$. The main result from~\cite{SaDo13} states that for $0\leq \epsilon \leq 1$, $C_1(\epsilon) = \epsilon$. All logarithms in this work are taken with base 2.

In this work, the single-neighbor constraint is naturally extended for two neighbors as it better captures the ICI phenomenon. This extension is applied both symmetrically and asymmetrically. In the symmetric version, as proposed in~\cite{SD13}, a permutation satisfies the constraint if the difference between the level of a cell and the level of one of its neighbors is bounded by some prescribed value $k$. In the asymmetric version, we will constrain the level difference only for sequences of the form high-low-high. This constraint is motivated by the fact that the ICI in flash memories mainly affects sequences of the form high-low-high and not the other ones. Thus, as in the single-neighbor constraint, we similarly define the capacity of these two constraints and show that in the symmetric constraint the capacity is $(1+\epsilon)/2$ and in the asymmetric constraint the capacity equals 1 regardless to the constraint value $\epsilon$.

The constraints studied in this paper as well as in~\cite{SaDo13} are affective in reducing the errors caused by the ICI. However, random errors may still happen. While there are several metrics under which error-correcting codes for permutations were studied, we choose to focus on the Kendall's $\tau$ metric due to its high applicability to the error behavior in the rank modulation scheme~\cite{JSB10}. Hence, we will study codes with minimum distance according to the Kendall's $\tau$ distance that yet consist of only permutations that satisfy the constraints studied in the paper.

The rest of the paper is organized as follows. In Section~\ref{sec:def}, we introduce the notations and formally define the constraints studied in the paper. Section~\ref{sec:2NC} studies the capacity of the symmetric constraint and Section~\ref{sec:AsymmConst} studies the capacity of the asymmetric constraint. In Section~\ref{sec:ECC}, we extend our results and study the capacity of these two constraints combined with error-correction codes in the Kendall's $\tau$ metric. Due to the lack of space, some of the proofs in the paper are omitted.

\section{Definitions and Notations}\label{sec:def}

In this section we formally define the constraints studied in the paper and introduce the notations and tools we will use in their solutions.
The set of $n$ elements $\{1,2,\ldots,n\}$ will be denoted by $[n]$. For two integers $a,b$, $a<b$, $[a,b]$ is the set of $b-a+1$ elements $[a,a+1,a+2,\ldots,b]$. Let $S_n$ be the set of all permutations on $[n]$, and let $S([a,b])$ be the set of all permutations on $[a,b]$. We denote a permutation $\sigma$ of length $n$ by $\sigma=[\sigma(1),\sigma(2),\ldots,\sigma(n)]$.
\begin{remark}
We use permutations of length $n$ in order to represent the rankings of $n$ flash memory cells in the rank modulation scheme. Note that there are two alternatives to represent the cells ranking by a permutation. The first one is the method we use in this paper where $\sigma(i)$ corresponds to the ranking of the $i$-th cell. In the second approach $\sigma(i)$ is the index of the cell with the $i$-th rank. While these two representations are dual to each other, we chose the first one for the convenience of describing the constraints in our work.
\end{remark}

\begin{definition}
Let $n$ and $k$ be positive integers such that $k<n$. A permutation $\sigma\in S_n$ is said to satisfy the \textbf{two-neighbor $k$-constraint} if for all $i$, $2\leq i\leq n-1$, either $|\sigma(i-1)-\sigma(i)|\leq k$ or $|\sigma(i)-\sigma(i+1)|\leq k$. We denote by $A_{n,k}$ the set of all permutations in $S_n$ satisfying the two-neighbor $k$-constraint.
A \textbf{two-neighbor $k$-constrained code} is a subset of $A_{n,k}$. Finally, for $0\leq \epsilon \leq 1$, the \textbf{capacity} of the two-neighbor $k$-constraint, where $k=\lceil n^{\epsilon}\rceil$, is
defined as
$$ C(\epsilon)=\limsup_{n\rightarrow \infty}\frac{\log |A_{n,k}|}{\log n!}. $$
\end{definition}
For example, the permutation $\sigma=[4,7,5,3,1,2,6]$ satisfies the two-neighbor 2-constraint but not
the two-neighbor $1$-constraint. Clearly, if $k = n-1$ then $A_{n,k}=S_n$.
Note that the two-neighbor constraint does not distinguish between high-low-high and low-high-low patterns and thus eliminates them both. A weaker constraint which may fit better to the inter-cell interference problem is defined next.
\begin{definition}
Let $n$ and $k$ be positive integers such that $k < n$. A permutation $\sigma\in S_n$ is said to satisfy the \textbf{asymmetric two-neighbor $k$-constraint} if for all $i$, $2\leq i\leq n-1$, either $\sigma(i-1)-\sigma(i)\leq k$ or $\sigma(i+1)-\sigma(i)\leq k$. The set of all permutations satisfying the asymmetric two-neighbor $k$-constraint is denoted by $B_{n,k}$. An \textbf{asymmetric two-neighbor $k$-constrained code} is a subset of $B_{n,k}$ and the constraint's capacity, for $0\leq \epsilon \leq 1$, is defined as $\widetilde{C}(\epsilon)=\limsup_{n\rightarrow \infty}\frac{\log |B_{n,k}|}{\log n!}$, where $k=\lceil n^{\epsilon}\rceil$.
\end{definition}
For example, the permutation $[5,3,1,6,4,2]$ satisfies the asymmetric two-neighbor $2$-constraint but not the asymmetric two-neighbor $1$-constraint. Note that every permutation which satisfies the two-neighbor $k$-constraint satisfies the asymmetric two-neighbor $k$-constraint as well and thus for any $0\leq \epsilon \leq 1$, $C(\epsilon)\leq \widetilde{C}(\epsilon)$.
\begin{remark}
We chose in the capacity definitions the supremum limit versions since the limits do not necessarily exist. However, we shall later see that these limits indeed exist.
\end{remark}

In the construction of two-neighbor $k$-constrained codes we will use some of the tools from multi-permutations, which are the natural generalization of permutations. A balanced multi-set $\cM_{\ell,m}=\{1^m,2^m,\ldots,\ell^m\}$ is a collection of the numbers in $[\ell]$, each appears $m$ times. The set of all multi-permutations over $\cM_{\ell,m}$ is denoted by $P_{\ell,m}$. This definition can be extended for multi-sets which are not balanced, however we will not need this generalization for our purposes. For a multi-permutation $\sigma\in P_{\ell,m}$, we distinguish between appearances of the same number in $\sigma$ according to their positions in $\sigma$. By abuse of notation, we sometimes write for $i\in[\ell], r\in[m]$, $\sigma(j)=i_r$ and $j=\sigma^{-1} (i_r)$ to indicate that the $r$-th appearance of $i$ is in the $j$-th position of $\sigma$. For example, if $\cM_{3,2}=\{1^2, 2^2, 3^2\}$ then $\sigma=[1,3,1,2,3,2]$ is a multi-permutation in $S_{3,2}$, and $\sigma(3)=1_2$, $3=\sigma^{-1}(1_2)$.

\section{The Two-Neighbor Constraint}\label{sec:2NC}
In this section we study the two-neighbor constraint and in particular find its capacity. This will be done first by a construction of two-neighbor $k$-constrained codes which provides a lower bound on the capacity. Then, we will show how to bound the size of the set $A_{n,k}$ which will result with an upper bound on the capacity that will coincide with the lower bound.

For a multi-permutation $\rho\in P_{\ell,m}$ and permutations $\gamma_1,\gamma_2,\ldots,\gamma_\ell$, such that $\gamma_i\in S([(i-1)m+1,im])$ for $i\in [\ell]$, we define $\rho(\gamma_1,\gamma_2,\ldots,\gamma_\ell)$ to be the permutation $\alpha\in S_{\ell m}$, such that $\alpha(j)=\gamma_i(r)$ if $\rho(j)=i_r$. For example, let $\rho=[1,2,1,3,2,3]\in P_{3,2}$ and let $\gamma_1=[2,1]$, $\gamma_2=[3,4]$ and $\gamma_3=[6,5]$. Then $\rho(\gamma_1,\gamma_2,\gamma_3)=[2,3,1,6,4,5]$. The following lemma will be useful in the construction we present in this section.
\begin{lemma}\label{lem:rhoInjection}
Let $\rho_1,\rho_2\in P_{\ell,m}$ and let $\gamma_1,\gamma_2,\ldots,\gamma_{\ell}$, $\delta_1,\delta_2,\ldots,\delta_{\ell}$, where $\gamma_i,\delta_i\in S([(i-1)m+1,im]$, for $i\in [\ell]$. If $\sigma=\rho_1(\gamma_1,\gamma_2,\ldots,\gamma_{\ell})=\rho_2(\delta_1,\delta_2,\ldots,\delta_{\ell})$, then $\rho_1=\rho_2$ and $\gamma_i=\delta_i$, for $i\in[\ell]$.
\end{lemma}

For an even integer $m$, the set $D_{\ell,m}\subseteq P_{\ell,m}$ is defined as follows. A multi-permutation $\rho\in P_{\ell,m}$ belongs to $D_{\ell,m}$ if for every $j$, $1\leq j\leq \ell m/2$, $\rho(2j-1)=\rho(2j)$. For example, the multi-permutation $\rho=[1,1,2,2,2,2,3,3,1,1,3,3]$ belongs to $D_{3,4}$ since $\rho(1)=\rho(2)$, $\rho(3)=\rho(4)$, and so on. The size of $D_{\ell,m}$ is equal to the size of $P_{\ell,m/2}$. In the next construction we show how to construct two-neighbor constrained codes.
\begin{construction}\label{cons:2NC}
Let $n=\ell(k+1)$, where $k$ is an odd positive integer and $\ell$ is a positive integer. Let $\cC_{n,k}^{sym}\subseteq S_n$ be the code consists of all the permutations $\sigma\in S_n$ of the form $\sigma=\rho(\gamma_1,\gamma_2,\ldots,\gamma_{\ell})$, where $\rho\in D_{\ell,k+1}$ and $\gamma_i\in S([(i-1)(k+1)+1,i(k+1)])$ for $i\in [\ell]$. That is,\vspace{-1ex}

\vspace{-1.5ex}
\begin{small}
$$
 \cC_{n,k}^{sym} =\left\{  \rho(\gamma_1,\ldots,\gamma_{\ell}) :\hspace{-1ex} \begin{array}{c}  \rho\in D_{\ell,k+1}, i\in [\ell],\\ \textrm{$\gamma_i\in S([(i-1)(k+1)+1,i(k+1)])$}\end{array}\hspace{-1ex}\right\}.\vspace{-1ex}
$$
\end{small}
\vspace{-3ex}

\end{construction}
The correctness of Construction~\ref{cons:2NC} as well as the code cardinality are proved in the next lemma.
\begin{lemma}\label{lem:construction}
Let $n,k,\ell$ be as specified in Construction~\ref{cons:2NC}. Then, the code $\cC_{n,k}^{sym}$ is a two-neighbor $k$-constrained code and its cardinality is\vspace{-1ex}
$$|\cC_{n,k}^{sym}|=\frac{\left(\frac{n}{2}\right)!(k+1)!^{\ell}}{\left(\frac{k+1}{2}\right)!^{\ell}}.\vspace{-1ex}$$
\end{lemma}
\begin{proof}
Let $\sigma\in \cC_{n,k}^{sym}$. Then there exist $\rho\in D_{\ell,k+1}$, and $\gamma_1,\gamma_2,\ldots,\gamma_{\ell}$, where $\gamma_i\in S([(i-1)(k+1)+1,i(k+1)])$, for $i\in [\ell]$, such that $\sigma=\rho(\gamma_1,\gamma_2,\ldots,\gamma_{\ell})$. Let $2 < j\leq n-1$ be an odd integer and assume that $\rho(j)=i_r$ for some $i\in [\ell]$ and $r\in [k+1]$. By the definition of $D_{\ell,k+1}$, it follows that $\rho(j+1)=i_{r+1}$. Hence, $\sigma(j)=\gamma_i(r)\in [(i-1)(k+1)+1,i(k+1)]$ and similarly $\sigma(j+1)=\gamma_i(r+1)\in [(i-1)(k+1)+1,i(k+1)]$. It follows that $|\sigma(j)-\sigma(j+1)|\leq k$. The case of $j$ even is handled the same with respect to the symbol in position $j-1$. Thus, $\sigma$ satisfies the two-neighbor $k$-constraint.

For the computation of the cardinality of $\cC_{n,k}^{sym}$, note that by Lemma \ref{lem:rhoInjection} it follows that every choice of $\rho\in D_{\ell,k+1}$ and $\gamma_1,\gamma_2,\ldots,\gamma_{\ell}$, where $\gamma_i\in  S([(i-1)(k+1)+1,i(k+1)])$, for $i\in [\ell]$, generates a different codeword of the form $\rho(\gamma_1,\gamma_2,\ldots,\gamma_{\ell})$. Therefore,\vspace{-2ex}
$$ |\cC_{n,k}^{sym}|=|D_{\ell,k+1}|\cdot(k+1)!^{\ell}=\frac{(\frac{n}{2})!(k+1)!^{\ell}}{(\frac{k+1}{2})!^{\ell}}.\vspace{-4ex}$$
\end{proof}

Even though Construction~\ref{cons:2NC} provides two-neighbor constrained codes only to the case where $k$ is odd, it can be easily modified for the case that $k$ is even as well. In any event, we will not need this modification in order to calculate a lower bound on the capacity, which is stated in the next theorem.
\begin{theorem}\label{th:lower bound}
For all $0\leq \epsilon\leq 1$, $C(\epsilon)\geq \frac{1+\epsilon}{2}$.
\end{theorem}

In order to derive an upper bound on the capacity $C(\epsilon)$ we show an upper bound on the size of $A_{n,k}$.
\begin{lemma}\label{lem:upper}
For all positive integers $n,k$ such that $k<n$,
$$|A_{n,k}|\leq 4^{n-1}k^{\frac{n}{2}}n^{\frac{n}{2}+1}.$$
\end{lemma}
\begin{proof}
Let $\psi:A_{n,k}\rightarrow \mathbb{Z}^n$ be the following mapping. For a permutation $\sigma\in A_{n,k}$,
$\psi(\sigma)=\mathbf{x}=(x_1,x_2,\ldots,x_n)\in \mathbb{Z}^n,$ where $x_1=\sigma(1)$, and for each $i$, $2\leq i\leq n$, $x_i=\sigma(i)-\sigma(i-1)$. Clearly, $\psi$ is an injection and therefore, the size of the set $A_{n,k}$ is equal to the size of the image of $\psi$, $\psi(A_{n,k})= \{\psi(\sigma) \ : \ \sigma\in A_{n,k}\}$. We will show an upper bound on the size of $\psi(A_{n,k})$.

Let $\mathbf{x}=\psi(\sigma)$ for some $\sigma\in A_{n,k}$. For any two consecutive positions $j,j+1$, $2\leq j\leq n-1$, either $|\sigma(j)-\sigma(j-1)|\leq k$ or $|\sigma(j+1)-\sigma(j)|\leq k$. Therefore, at least $\left\lfloor \frac{n-1}{2}\right\rfloor$ of the $n-1$ elements $x_2,x_3,\ldots,x_n$ are in the range $[-k,k]\setminus\{0\}$. Let $I\subseteq [2,n]$ be a set with at least $\left\lfloor \frac{n-1}{2}\right\rfloor$ elements and let $D_I$ be the set of all vectors $\mathbf{x}\in \psi(A_{n,k})$ for which $x_i\in [-k,k]\setminus\{0\}$, for every $i\in I$ and $x_j\in [-n,n]\setminus [-k,k]$, for every $j\in [2,n]\setminus I$. Then,
\begin{equation}\label{eq:1}
|\psi(A_{n,k})|\leq \sum_{I\subseteq[2,n],~|I|\geq \left\lfloor \frac{n-1}{2}\right\rfloor}|D_I|.
\end{equation}

For each $i\in I$ there are $2k$ choices for $x_i$ and for each $j\in [2,n]\setminus I$ there are at most $2(n-k)<2n$ choices for $x_j$. Finally, there are $n$ choices for $x_1$. Therefore,
$$|D_I|\leq n\cdot (2k)^{\left\lfloor \frac{n-1}{2}\right\rfloor}\cdot (2n)^{\left\lceil \frac{n-1}{2}\right\rceil}= 2^{n-1}k^{\left\lfloor \frac{n-1}{2}\right\rfloor}n^{\left\lceil\frac{n-1}{2}\right\rceil+1}.$$

Since the number of choices of $I$ is less than $2^{n-1}$, according to~(\ref{eq:1}), the following upper bound on the cardinality of $A_{n,k}$ and $\psi(A_{n,k})$ is derived\vspace{-1ex}
\begin{align*}
& |A_{n,k}| = |\psi(A_{n,k})| \leq 2^{n-1}\cdot 2^{n-1}k^{\left\lfloor \frac{n-1}{2}\right\rfloor}n^{\left\lceil\frac{n-1}{2}\right\rceil+1}& \\
& \leq  4^{n-1}k^{\frac{n}{2}}n^{\frac{n}{2}+1}. &
\end{align*}\vspace{-2ex}
\end{proof}

As a result of the last lemma we derive the following.
\begin{theorem}\label{th:upper bound}
For all $0\leq \epsilon\leq 1$, $C(\epsilon)\leq \frac{1+\epsilon}{2}$.
\end{theorem}

The following Corollary, which is an immediate result of Theorems~\ref{th:lower bound} and~\ref{th:upper bound}, summarizes the section's discussion.
\begin{cor}\label{cor:2NC}
For all $0\leq \epsilon\leq 1$, $ C(\epsilon)=\frac{1+\epsilon}{2}$.
\end{cor}

\section{The Asymmetric Two-Neighbor Constraint}\label{sec:AsymmConst}

In this section we find the capacity of the asymmetric two-neighbor constraint. Our main result states that for all $0\leq \epsilon \leq 1$, $\widetilde{C}(\epsilon)=1$. Since the capacity is at most 1, and the capacity is nondecreasing when $\epsilon$ increases, we will need to show that $\widetilde{C}(0)=1$. This will be done by a construction of an asymmetric two-neighbor $1$-constrained code that confirms this capacity result.

For a set $I$, let $I^{\nearrow}$, respectively $ I^{\searrow}$, denote the ordering of all elements in $I$ according to their increasing, respectively decreasing, order. For the construction of an asymmetric two-neighbor $1$-constrained code we will need the code $\cC_{r',1}^{sym}$, where $r'$ is even, from Construction \ref{cons:2NC}. Recall that a permutation $\pi\in \cC_{r',1}^{sym}$ is of the form\vspace{-1ex}
$$
\pi=\rho(\gamma_1,\gamma_2,\ldots,\gamma_{\frac{r'}{2}}),\vspace{-1ex}
$$
where $\rho(2i-1)=\rho(2i)$ and
$\gamma_i\in S([2i-1,2i])$, for all $1\leq i\leq \frac{r'}{2}$.
In other words, for every $j$, $1\leq j\leq \frac{r'}{2}$, there exists $1\leq i\leq \frac{r'}{2}$ such that $\{\pi(2j-1),\pi(2j)\}=\{2i-1,2i\}$.

\begin{construction}
\label{con:ATNC}
Let $r$ be an integer, $1\leq r\leq \frac{n}{2}$. If $r$ is even,
let the code $\cC_r\subset S_n$ be defined as follows. A permutation $\sigma\in S_n$ belongs to $\cC_r$ if there exists a partition of the set $[r-1,n]$ into $r$ nonempty sets $I_1,I_2,\ldots,I_r$, and a permutation $\pi\in \cC_{r-2,1}^{sym}$ such that

\vspace{-2ex}
\begin{small}
$$
\hspace{-0.1ex}\sigma\hspace{-0.6ex}=\hspace{-0.3ex}[I_1^{\nearrow}\hspace{-0.3ex},\hspace{-0.3ex}I_2^{\searrow}\hspace{-0.3ex},\hspace{-0.3ex}\pi(1),\hspace{-0.3ex}\pi(2), \hspace{-0.3ex}I_3^{\nearrow}\hspace{-0.3ex},\hspace{-0.3ex}I_4^{\searrow}\hspace{-0.3ex},\hspace{-0.3ex}\ldots\hspace{-0.3ex},\hspace{-0.3ex}\pi(r-3), \hspace{-0.3ex}\pi(r-2),I_{r-1}^\nearrow\hspace{-0.3ex},I_r^\searrow\hspace{-0.3ex}].\vspace{-2.5ex}
$$
\end{small}

For an odd $r$, let the code $\cC_r\subset S_n$ defined in a similar way. A permutation $\sigma\in S_n$ belongs to $\cC_r$ if there exists a partition of the set $[r,n]$ into $r$ nonempty sets $I_1,I_2,\ldots,I_r$, and a permutation $\pi\in \cC_{r-1,1}^{sym}$ such that

\vspace{-2ex}
\begin{small}
$$\sigma=[I_1^\nearrow,I_2^\searrow,\pi(1),\pi(2),I_3^\nearrow,I_4^\searrow,\ldots,\pi(r-2),\pi(r-1),I_{r}^\nearrow].\vspace{-2.5ex}$$
\end{small}

Finally, let $\cC_n^{asym}\subset S_n$ be the code\vspace{-2ex}
$$\cC_n^{asym}=\bigcup_{r=1}^{\lfloor n/2\rfloor}\cC_r.\vspace{-1ex}$$
\end{construction}

\begin{example}
For $n=14$ and $r=5$, let $I_1=\{5,8,10\}$, $I_2=\{6,12\}$, $I_3=\{7,15\}$, $I_4=\{9,13\}$, $I_5=\{11,14\}$ be a partition of $[5,14]$ into $5$ nonempty sets and let $\pi=[4,3,1,2]$. Note, that $\pi=\rho(\gamma_1,\gamma_2)$ where $\rho=[2,2,1,1]$, $\gamma_1=[1,2]\in S([1,2])$, and $\gamma_2=[4,3]\in S([3,4])$, hence, $\pi$ is a codeword in $\cC_{4,1}^{sym}$. The permutation $\sigma\in \cC_5$ of the form \begin{small}$\sigma=[I_1^{\nearrow},I_2^{\searrow},\pi(1),\pi(2),I_3^{\nearrow},I_4^{\searrow},\pi(3),\pi(4),I_5^{\nearrow}]$
\end{small} is $\sigma=[5,8,10,12,6,4,3,7,15,13,9,1,2,11,14]$.
Note, that $\sigma$ can also be obtained from other partitions such as $\tilde{I}_1=\{5,8,10,12\}$, $\tilde{I}_2=\{6\}$, and $\tilde{I}_i=I_i$, for all $3\leq i\leq 5$.
\end{example}

A position $i$, $2\leq i\leq n-1$, is called a \textbf{valley} in a permutation $\sigma\in S_n$ if
$\sigma(i-1)>\sigma(i)$ and $\sigma(i)<\sigma(i+1)$. For example, in the permutation $\sigma=[4,7,5,6,1,2,3]$, the third  and fifth positions are valleys. The next lemma will be used in proving the correctness of the construction, which will be proved next.
\begin{lemma}\label{lem:CrCupCr+1}
Let $m$ be an integer, $0\leq m\leq \frac{n-2}{4}$. Then, every permutation $\sigma\in\cC_{2m+1}\cup\cC_{2m+2}$ has exactly $m$ valleys.
\end{lemma}

\begin{lemma}
\label{lem:A2NC1}
For all $n\geq 1$, the code $\cC_n^{asym}$ is an asymmetric two-neighbor $1$-constrained code.
\end{lemma}
\begin{proof}
Let $\sigma\in \cC_n^{asym}$ and let $m$ be the number of valleys in $\sigma$. By Lemma \ref{lem:CrCupCr+1} it follows that $\sigma\in \cC_{2m+1}\cup\cC_{2m+2}$ and the valleys of $\sigma$ are the positions $i$ such that $\sigma(i)=\pi(j)$, for some $1\leq j\leq 2m$, and $\pi(j)$ is odd. It follows that either $\sigma(i-1)=\sigma(i)+1$ or $\sigma(i+1)=\sigma(i)+1$. Then the valleys in $\sigma$ do not violate the asymmetric two-neighbor 1-constraint and therefore $\sigma$ satisfies the asymmetric two-neighbor 1-constraint.
\end{proof}

Next, we will analyze a lower bound on the cardinalities of the codes from Construction~\ref{con:ATNC}. First, we use the following observation.
\begin{lemma}
\label{lem:A2NC2}
For all $n\geq 1$, let $\sigma \in \cC_n^{asym}$ and let $m$ be the number of valleys in $\sigma$. Then there exist at most $2^{m+1}$ different ways to obtain $\sigma$ as described in Construction~\ref{con:ATNC}.
\end{lemma}

For two positive integers $\ell,r$, where $r\leq \ell$, the number of partitions of $\ell$ elements into $r$ nonempty sets is denoted by $S(\ell,r)$ and is known as the Stirling number of the second kind.
\begin{lemma}\label{lem:A2NCcard}
For all $n\geq 1$, the cardinality of the code $\cC_n^{asym}$ satisfies\vspace{-1.5ex}
$$ |\cC_n^{asym}|\geq\sum_{r=1}^{\lfloor\frac{n}{2}\rfloor}\frac{1}{2}r!S\Big(n-2\Big\lfloor \frac{r-1}{2}\Big\rfloor,r\Big)\Big\lfloor\frac{r-1}{2}\Big\rfloor!.\vspace{-1ex}$$
\end{lemma}
\begin{proof}
For every $m$, $0\leq m\leq \frac{n-2}{4}$, we compute a lower bound on the size of $\cC_{2m+1}\cup\cC_{2m+2}$. There are $r!S(n-2m,r)$ choices for the partition $I_1,I_2,\ldots,I_r$, where $r=2m+1$ or $r=2m+2$, and there are $m!\cdot 2^{m}$ choices for the permutation $\pi\in \cC_{2m,1}$. By Lemma \ref{lem:A2NC2} the expression

\vspace{-2ex}
\begin{small}
$$[(2m+1)!S(n-2m,2m+1)+(2m+2)!S(n-2m,2m+2)]m!2^m \vspace{-3ex}$$
\end{small}

\noindent counts codewords in $\cC_{2m+1}\cup\cC_{2m+2}$ and each codeword in $\cC_{2m+1}\cup\cC_{2m+2}$ is counted at most $2^{m+1}$ times.
Hence, the size of $\cC_{2m+1}\cup\cC_{2m+2}$ is at least

\vspace{-2ex}
\begin{small}
$$[(2m+1)!S(n-2m,2m+1)+(2m+2)!S(n-2m,2m+2)]\frac{m!}{2}.\vspace{-3ex}$$
\end{small}

By Lemma \ref{lem:CrCupCr+1} it follows that the sets $\cC_{2m+1}\cup\cC_{2m+2}$ and $\cC_{2m'+1}\cup\cC_{2m'+2}$ are disjoint if $m'\neq m$, and therefore\vspace{-1.5ex}
$$ |\cC_n^{asym}|\geq\sum_{r=1}^{\lfloor\frac{n}{2}\rfloor}\frac{1}{2}r!S\Big(n-2\Big\lfloor \frac{r-1}{2}\Big\rfloor,r\Big)\Big\lfloor\frac{r-1}{2}\Big\rfloor!.\vspace{-3ex}$$
\end{proof}

Finally, the next theorem, which is a direct result of Lemma~\ref{lem:A2NCcard} and a lower bound on the Stirling numbers of the second kind, highlights the result of this section.
\begin{theorem}
\label{thm:A2NCcap}
For all $0\leq \epsilon\leq 1$, $\widetilde{C}(\epsilon)=1$.
\end{theorem}

\section{The Capacity of Error-Correcting Constrained Codes}\label{sec:ECC}
The two-neighbor constraint and the asymmetric two-neighbor constraint were proposed to combat errors that are caused by the inter-cell interference in flash memory cells. However, constrained codes should also be restricted to have error-correction capabilities, which is the topic of this section. A similar problem for the one-neighbor constraint was studied in~\cite{SD13}.

Given a permutation $\sigma=[\sigma(1),\sigma(2),\ldots,\sigma(n)]\in S_n$, an \emph{adjacent transposition} is an exchange of two adjacent elements
$\sigma(i),\sigma(i+1)$, in $\sigma$, for some $1\leq i\leq n-1$. The result of such an adjacent transposition is the permutation $[\sigma(1),\ldots,\sigma(i-1),\sigma(i+1),\sigma(i),\sigma(i+2),\ldots,\sigma(n)]$. The \emph{Kendall's $\tau$ distance} between two permutations
$\sigma,\pi\in S_n$, denoted by $d_{K} (\sigma,\pi)$, is the minimum number of adjacent transpositions required to obtain the permutation $\pi$ from the permutation $\sigma$.

For two permutations $\sigma,\pi\in S_n$, the \emph{inversion distance}, denoted by $d_{I}(\sigma,\pi)$, between $\sigma$ and $\pi$ is the Kendall's $\tau$ distance between their inverses, i.e.,\vspace{-1ex}
$$d_I(\sigma,\pi)=d_K(\sigma^{-1},\pi^{-1}).\vspace{-1ex}$$
Even though this distance was studied before, see e.g.~\cite{DiGr77}, we are not aware of any formal name for this metric and thus call it here the inversion distance. In this section we study the capacity of the constraints in this paper combined with a requirement of a minimum inversion distance.
\begin{remark}
We study the inversion distance and not the Kendall's $\tau$ one since, according to our representation of the cells ranking in a permutation, this metric fits better with the error behavior in flash memory cells. The motivation in studying codes in the Kendall's $\tau$ metric originated from the observation that cells with adjacent levels may interchange their rankings~\cite{JSB10}. Therefore, codes in the Kendall's $\tau$ metric should be invoked over the inverses of the permutations. However, in order to study these codes with constrained codes, one should take the inversion distance applied for the permutations.
\end{remark}

Let $E(n,k,d)$ be the maximum size of a code in $A_{n,k}$ with minimum inversion distance $d$. For $0\leq\epsilon_1 \leq 1$ and $0\leq \epsilon_2\leq 2$, let $k=\lceil n^{\epsilon_1}\rceil$ and $d=\lceil n^{\epsilon_2}\rceil$, and define the capacity of two-neighbor $k$-constrained codes with minimum inversion distance $d$ by
$$C(\epsilon_1,\epsilon_2)=\lim_{n\rightarrow \infty } \frac{\log E(n,k,d)}{\log n!}.$$

We will compute this capacity in terms of $\epsilon_1$ and $\epsilon_2$ by following some of the methods used in \cite{BM10} and later in \cite{SaDo13}. We distinguish between three cases:
\begin{enumerate}
\item $0\leq \epsilon_2\leq 1$ and $0\leq \epsilon_1\leq 1$,
\item $1<\epsilon_2\leq 1+\epsilon_1$, and $0\leq \epsilon_1\leq 1$,
\item $1+\epsilon_1<\epsilon_2\leq 2$ and $0\leq \epsilon_1\leq 1$.
\end{enumerate}

For $\sigma\in S_n$, the \emph{ball} in $S_n$ of radius $r$ centered at $\sigma$ is defined by\vspace{-1ex}
$$\mathcal{B}_I(n,\sigma,r)\deff\{\pi\in S_n~:~d_I(\sigma,\pi)\leq r\}.\vspace{-1ex}$$
The size of the ball $\mathcal{B}_I(n,\sigma,r)$ does not depend on $\sigma$ and thus we denote it by $b_I(n,r)$. For $\sigma\in A_{n,k}$, the \emph{ball} in $A_{n,k}$ of radius $r$ centered at $\sigma$ is defined by\vspace{-1ex}
$$\mathcal{B}_I(A_{n,k},\sigma,r)\deff\{\pi\in A_{n,k}~:~d_K(\sigma,\pi)\leq r\}.\vspace{-1ex}$$

A code in $A_{n,k}$ with minimum inversion distance $d$ can be constructed by a greedy approach which leads to the following Gilbert-Varshamov type of lower bound.
\begin{lemma}
\label{lem:upLowBound}
For every $1\leq k<n$, $1\leq d \leq {n\choose 2}$, the following lower bound on $E(n,k,d)$ holds\vspace{-1ex}
$$ E(n,k,d) \geq \frac{|A_{n,k}|}{b_I(n,d-1)}.\vspace{-1ex}$$
\end{lemma}

The next theorem is a combination of results from \cite{BM10}, \cite{LoPr03}, and \cite{Mar01}.
\begin{theorem}
\label{thm:KsphereUpBound} Let $r=\Theta(n^{\delta})$, where $0\leq \delta \leq 2$. Then there exist constants $c_1$ and $c_2$ such that\vspace{-1ex}
$$
b_I(n,r)\leq \begin{cases}
e^{c_1n} ,& 0\leq \delta\leq 1,\\
(c_2n^{\delta-1})^n,  & 1<\delta \leq 2.
 \end{cases}\vspace{-1ex}
$$
\end{theorem}

We are now in a position to compute the capacity $C(\epsilon_1,\epsilon_2)$ for the first case.
\begin{theorem}
\label{thm:capLeq1}
For ${0\leq \epsilon_1, \epsilon_2\leq 1}$, $ C(\epsilon_1,\epsilon_2)=\frac{1}{2}+\frac{\epsilon_1}{2}$.
\end{theorem}
\begin{proof}
Since $E(n,k,d)\subseteq A_{n,k}$ it follows that\vspace{-1ex}
$$\frac{\log E(n,k,d)}{\log n!}\leq \frac{\log |A_{n,k}|}{\log n!},\vspace{-1ex}$$
and hence from Corollary~\ref{cor:2NC}, $ C(\epsilon_1,\epsilon_2)\leq C(\epsilon_1)=\frac{1}{2}+\frac{\epsilon_1}{2}$.

By Lemma \ref{lem:upLowBound} and Theorem \ref{thm:KsphereUpBound} there exists a constant $c$ such that\vspace{-1ex}
$$\frac{\log E(n,k,d)}{\log n!}\geq \frac{\log |A_{n,k}|}{\log n!}-\frac{\log e^{cn}}{\log n!}.\vspace{-1ex}$$
Then, $C(\epsilon_1,\epsilon_2)\geq C(\epsilon_1)=\frac{1}{2}+\frac{\epsilon_1}{2},$ and thus, $C(\epsilon_1,\epsilon_2)=\frac{1}{2}+\frac{\epsilon_1}{2}$.
\end{proof}

Before proceeding to the second case, let us introduce some more tools that we will use in solving this case. Let $H_n=\{1,2,\ldots,n\}^n$. For $\mathbf{x},\mathbf{y}\in H_n$, the Manhattan distance between $\mathbf{x}$ and $\mathbf{y}$,
$d_M(\mathbf{x},\mathbf{y})$, is defined as\vspace{-1.5ex}
$$d_M(\mathbf{x},\mathbf{y})\deff\sum_{i=1}^n|x_i-y_i|.\vspace{-1ex}$$
The next lemma was proved in \cite{DiGr77}.
\begin{lemma}
\label{lem:metrics}
For every $\sigma,\pi\in S_n$,\vspace{-1ex}
$$\frac{1}{2}d_M(\sigma,\pi)\leq d_I(\sigma,\pi)\leq d_M(\sigma,\pi).\vspace{-1ex}$$
\end{lemma}

The definition of the two-neighbor $k$-constraint can be trivially extended to $H_n$. A vector $\mathbf{x}\in H_n$ satisfies the two-neighbor $k$-constraint if either $|x_i-x_{i-1}|\leq k$ or $|x_{i+1}-x_i|\leq k$, for all $2\leq i\leq n-1$.
Let $\mathcal{A}_{n,k}$ be the set of all elements of $H_n$ that satisfy the two-neighbor $k$-constraint.

For a subset $S\subseteq H_n$ and $\mathbf{x}\in S$, the Manhattan \emph{ball} in $S$ of radius $r$ centered at $\mathbf{x}$ is defined by\vspace{-1ex}
$$\mathcal{B}_M(S,\mathbf{x},r)\deff\{\mathbf{y}\in S~:~d_M(\mathbf{x},\mathbf{y})\leq r\}.\vspace{-1ex}$$

Combining the previous results along with the sphere packing upper bound and Gilbert-Varshamov lower bound provides us with the following lemma.
\begin{lemma}\label{lem:ElowUp}
For every $1\leq k<n$, $1\leq d \leq {n\choose 2}$,\vspace{-1.5ex}
$$E(n,k,d)\leq  \frac{|\mathcal{A}_{n,k}|}{\min_{\mathbf{x}\in \mathcal{A}_{n,k}}\{|\mathcal{B}_M(\mathcal{A}_{n,k},\mathbf{x},\left\lfloor\frac{d-1}{2}\right\rfloor)|\}}.\vspace{-1.5ex}$$
and\vspace{-1.5ex}
$$E(n,k,d)\geq \frac{|A_{n,k}|}{\max_{\mathbf{x}\in \mathcal{A}_{n,k}}\{|\mathcal{B}_M(\mathcal{A}_{n,k},\mathbf{x},2d-1)|\}}.\vspace{-1ex}$$
\end{lemma}

In order to apply the upper bound from Lemma~\ref{lem:ElowUp}, we state in the next lemma a lower bound on the size of a ball in the Manhattan distance in $\mathcal{A}_{n,k}$.
\begin{lemma}\label{lem:BMAlower}
Let $k=\lceil n^{\epsilon}\rceil$ and $r=\lceil n^{\delta}\rceil $, where $0\leq \epsilon \leq 1, 0\leq \delta \leq 2$.
Then there exists a constant $c$ such that\vspace{-1ex}
$$
\min_{\mathbf{x}\in \mathcal{A}_{n,k}}\hspace{-1ex}\{|\mathcal{B}_M(\mathcal{A}_{n,k},\mathbf{x},r)|\}\geq \begin{cases}
\left(\frac{n^{\delta-1}}{2}\right)^n &\hspace{-2ex}, 1<\delta<1+\epsilon<2\\
\left(\frac{n^{\delta-1+\epsilon}}{c}\right)^{\frac{n}{2}} &\hspace{-2ex}, 1+\epsilon\leq \delta<  2
\end{cases}\vspace{-1ex}
$$
\end{lemma}

We are ready to prove the capacity for the second case.
\begin{theorem}
\label{thm:cap1toEpsilon}
For $0\leq\epsilon_1\leq 1$ and $1<\epsilon_2\leq 1+\epsilon_1$,\vspace{-1ex}
$$ C(\epsilon_1,\epsilon_2)=\frac{3}{2}+\frac{\epsilon_1}{2}-{\epsilon_2}.\vspace{-1ex}$$
\end{theorem}
\begin{proof} Let $k=\lceil n^{\epsilon_1}\rceil$ and $d=\lceil n^{\epsilon_2}\rceil$.
By Lemma \ref{lem:upLowBound} and Theorem \ref{thm:KsphereUpBound} it follows that there exists a constant $c$ such that\vspace{-1ex}
$$\frac{\log E(n,k,d)}{\log n!}\geq \frac{\log |A_{n,k}|}{\log n!}-\frac{\log c^nn^{(\epsilon_2-1)n}}{\log n!}.\vspace{-1ex}$$
Therefore,\vspace{-1ex}
$$ C(\epsilon_1,\epsilon_2)\geq \frac{1}{2}+\frac{\epsilon_1}{2}+1-\epsilon_2=\frac{3}{2}+\frac{\epsilon_1}{2}-\epsilon_2.\vspace{-1ex}$$

Similarly, by Lemmas \ref{lem:ElowUp} and \ref{lem:BMAlower} it follows that\vspace{-1ex}
$$\frac{\log E(n,k,d)}{\log n!}\leq \frac{\log |\mathcal{A}_{n,k}|}{\log n!}-\frac{\log \left(\frac{n^{\epsilon_2-1}}{2}\right)^n}{\log n!},\vspace{-1ex}$$
and hence,\vspace{-1ex}
$$C(\epsilon_1,\epsilon_2)\leq \frac{1}{2}+\frac{\epsilon_1}{2}+1-\epsilon_2.\vspace{-1ex}$$
Together we get, $C(\epsilon_1,\epsilon_2)=\frac{3}{2}+\frac{\epsilon_1}{2}-\frac{\epsilon_2}{2}$.
\end{proof}

Due to the lack of space we skip the details of the third case and thus summarize with the following corollary.
\begin{cor}
Let $0\leq \epsilon_1\leq 1$ and $0\leq \epsilon_2\leq 2$. Then\vspace{-1ex}
$$
C(\epsilon_1,\epsilon_2)= \begin{cases}
\frac{1}{2}+\frac{\epsilon_1}{2},& 0\leq \epsilon_2\leq 1,\\
\frac{3}{2}+\frac{\epsilon_1}{2}-{\epsilon_2} ,& 1<\epsilon_2\leq 1+\epsilon_1,\\
1-\frac{\epsilon_2}{2} ,& 1+\epsilon_1< \epsilon_2\leq 2.
\end{cases}\vspace{-1ex}
$$
\end{cor}

Lastly, we report on similar results for the asymmetric constraint. Let $\widetilde{C}(\epsilon_1,\epsilon_2)$ be the capacity of an asymmetric two neighbor $k$-constrained code with minimum inversion distance $d$, where $k=\lceil n^{\epsilon_1}\rceil$ and $d=\lceil n^{\epsilon_2} \rceil$. Following the same technique used in \cite{BM10} we have
\begin{theorem}
Let $0\leq \epsilon_1\leq 1$ and $0\leq \epsilon_2 \leq 2$. Then\vspace{-1ex}
$$
\widetilde{C}(\epsilon_1,\epsilon_2) \begin{cases}
1,& 0\leq \epsilon_2\leq 1,\\
2-\epsilon_2 ,& 1< \epsilon_2\leq 2.
\end{cases}\vspace{-1ex}
$$
\end{theorem}

\section*{Acknowledgment}
The work of Sarit Buzaglo was supported in part by the U.S.-Israel Binational Science Foundation, Jerusalem, Israel, under Grant No. 2012016.
The work of Eitan Yaakobi was supported in part by Intellectual Ventures, an NSF grant CIF-1218005, and the U.S.-Israel Binational Science Foundation, Jerusalem, Israel, under Grant No. 2010075.

\end{document}